\documentclass[a4paper,onecolumn,superscriptaddress,11pt,accepted=2022-06-14]{quantumarticle}
\pdfoutput=1 

\usepackage{amsfonts}
\usepackage{amsmath}
\usepackage{amssymb}
\usepackage{hyperref}
\usepackage{amsthm}
\usepackage[mathscr]{eucal}
\usepackage{bbold}
\usepackage{braket}
\usepackage{color}
\usepackage{dsfont}
\usepackage{enumitem}
\usepackage{mathtools}
\usepackage{physics}
\usepackage{tensor}
\usepackage[normalem]{ulem}
\usepackage{ulem}
\usepackage{comment}
\usepackage[ruled,vlined]{algorithm2e}

\newcommand{\ignore}[1]{}%

\numberwithin{equation}{section}

\newtheorem{nthm}{Theorem}[]
\newtheorem{nlemma}{Lemma}
\newtheorem{nprop}{Proposition}

\newtheorem{ncor}{Corollary}
\newtheorem{ndefi}{Definition}

\definecolor{shccolor}{rgb}{0.9,0.1,0.1}

\definecolor{vpscolor}{rgb}{0.1,0.4,0.9}

\definecolor{ncccolor}{rgb}{0.4,0.7,0.4}

\definecolor{nbcolor}{rgb}{0.5,0.3,0.4}

\definecolor{updatecolor}{rgb}{1,0,0}
\newcommand{\update}[1]{#1}

\allowdisplaybreaks
\makeatletter
\def\@fpheader{ }
\makeatother

\begin{document}

\title{Topological link models of multipartite entanglement}

\author[a]{Ning Bao}
\email{ningbao75@gmail.com}
\author[b]{Newton Cheng}
\email{newtoncheng@berkeley.edu}
\author[c]{Sergio Hern\'andez-Cuenca}
\email{sergiohc@ucsb.edu}
\author[b]{Vincent Paul Su}
\email{vipasu@berkeley.edu}

\affiliation[a]{Computational Science Initiative, Brookhaven National Lab, Upton, NY, 11973, USA}
\affiliation[b]{Center for Theoretical Physics, Department of Physics, University of California, Berkeley, CA 94720, USA}
\affiliation[c]{Department of Physics, University of California, Santa Barbara, CA 93106, USA}

\begin{abstract}
    We introduce a novel model of multipartite entanglement based on topological links, generalizing the graph/hypergraph entropy cone program. We demonstrate that there exist link representations of entropy vectors which provably cannot be represented by graphs or hypergraphs. Furthermore, we show that the contraction map proof method generalizes to the topological setting, though now requiring oracular solutions to well-known but difficult problems in knot theory.
\end{abstract}

\maketitle

\section{Introduction}
Originally motivated by relations obeyed by the entropy function in holography~\cite{Ryu:2006bv},~\cite{Bao:2015bfa,HernandezCuenca:2019wgh,Avis:2021xnz} studied a graph model of entanglement, wherein the von Neumann entropies of quantum subsystems are calculated by minimum cuts separating those subsystems from their complements in a corresponding weighted graph. The combinatorial properties of minimum cuts on graphs allowed for the development of a method to prove entropy inequalities obeyed by holographic states. While powerful, this technology was limited to graph models intended to describe the subset of holographic states, and therefore did not capture the entanglement structure of more general quantum states. In follow-up work, the graph model was generalized to a hypergraph model~\cite{Bao:2020zgx}, with the aim of reproducing the entropies of a wider class of states while retaining the ability to prove entropy inequalities via a generalization of the proof technology. For four parties, this was indeed the case: the hypergraph cone coincided with that of stabilizer states, which is larger than the holographic one of graphs, but still not large enough to capture general quantum states.\footnote{The full four-party quantum entropy cone is unknown~\cite{pippenger2003inequalities}, though it is known to properly contain the stabilizer entropy cone as defined in~\cite{Linden:2013kal}.}

The hypergraph model also provided interesting new results intimately related to the properties of ranks of linear subspaces and the entropies of stabilizer states in quantum mechanics~\cite{Bao:2020zgx}. Further explorations of their connection to the latter revealed that the hypergraph model was able to encode the entropies of a subset of stabilizer states~\cite{Walter:2020zvt,Nezami:2016zni}, and that this containment was proper~\cite{Bao:2020mqq}. The latter result relied on the discovery of a particular stabilizer state whose entropy vector, referred to as ray $15$ in~\cite{Bao:2020mqq}, violates an inequality proven to be valid for all hypergraphs in the hypergraph model.

By contrast, another avenue of study has been the study of topology as an efficient encoding of certain patterns of entanglement. Some work in this direction include~\cite{Salton:2016qpp,Balasubramanian:2016sro,Chun:2017hja,Mironov:2019taf,Kauffman:2019,Aharonov2008APQ}. Because all graph and hypergraph models of entanglement can be captured by topological links (as will be discussed in detail later in this work), using topological links to generalize encodings of entanglement patterns in a systematic way is an intriguing avenue of exploration to describe the entanglement of more general quantum states.

In this paper, we introduce such a generalization of the graph and hypergraph models of quantum entanglement based on topological links. \update{Similar to the (hyper)graph case, entropies will be given by the weights of minimal cuts on links, and used to define the \textit{entropy cone of link models} as the set of all entropy vectors realizable by such models. In Section \ref{sec:prelim}, we review the graph and hypergraph models and introduce the contraction map method that was applied to prove entropy inequalities on such models. In Section \ref{sec:link models}, we then formally introduce link models as a generalization of hypergraph models. We demonstrate that such models strictly contain the hypergraph entropy cone by exhibiting a link whose entropies capture those of ray 15. We then prove a generalized contraction map condition on such models that, subject to access to an oracle whose abilities we carefully define, provides a sufficient condition for the validity of a given entropy inequality on a link model. We also motivate that link models may usefully capture the entropies of quantum states by showing that the entropy cones of link models and quantum states are identical for $n \leq 3$ parties. A more detailed discussion of how link models can be applied to further characterize entanglement structures of quantum states is provided in Section \ref{sec:disc}}.

\section{Preliminaries}\label{sec:prelim}
\subsection{Graph models}
We first review the definition of the graph model and the intuition behind it, as originally introduced and studied in~\cite{Bao:2015bfa}. For any positive integer $k$, let $[k]=\{1,\dots,k\}$.
A graph model consists of an undirected graph $G = (V,E)$ with nonnegative edge weights specified by a map $w : E \to \mathbb{R}_{\ge0}$. A subset of vertices $\partial V\subseteq V$ are called \textit{external vertices} and the rest are referred to as \textit{internal vertices}. For $\partial V$ of cardinality $n+1$, a bijective coloring $\partial V \to [n+1]$ is used to identify internal vertices with partitions of a quantum state: $G$ represents the entropies of a mixed $n$-party state on the vertices $[n]$, with vertex $n+1$ representing the purification of the state, and all internal vertices encoding the entanglement structure. In this setting, there are $2^n-1$ subsystems of the quantum state one can consider, corresponding to every non-empty subset $I\subseteq[n]$. The entropy of any one of these is defined as the weight of the minimum cut or \textit{min-cut} separating vertices $I$ from $\partial V \smallsetminus I$.

Given an arbitrary graph model, one may intuitively expect that the min-cut prescription should place restrictions on the entropies of states that can be faithfully represented by such models. These restrictions turn out to take the form of linear entropy inequalities, which can be canonically written as
\begin{equation}\label{eq:candidate_ineq}
    \sum_{l=1}^L \alpha_l S(I_l) \geq \sum_{r=1}^R \beta_r S(J_r),
\end{equation}
where $\alpha_l,\beta_r > 0$, and $I_l,J_r$ are subsets of external vertices. In other words, we arrange terms into a left-hand side (LHS) and a right-hand side (RHS) such that they all appear with positive coefficients. The entropies of any state which can be represented by a graph model will obey inequalities of the general form of \eqref{eq:candidate_ineq}.

These inequalities on graph models can be proven via a technique known as proof-by-contraction. The method functions by essentially formalizing the intuition that the weights of any clever construction of cuts corresponding to a given choice of external vertices are always lower bounded by the weights of the corresponding min-cut.\footnote{See~\cite{Akers:2021lms} for a description of the contraction-map proof method in the geometric setting that originally inspired it, based on the holographic Ryu-Takayanagi formula~\cite{Ryu:2006bv}.} \update{To give the precise statement of the proof-by-contraction method, let us first introduce some basic notation.
For each $i\in[n+1]$, let $x^{(i)}$ and $y^{(i)}$ respectively be bitstrings of lengths $L$ and $R$, with entries given by $x^{(i)}_l = \delta(i \in I_l)$ and $y^{(i)}_r = \delta(i \in J_r)$, where $\delta$ is a Boolean indicator function yielding 1 or 0 depending on whether its argument is true or false, respectively. These are often referred to as \emph{occurrence bitstrings}, as they simply encode the occurrence of the $i^{\text{th}}$ party in $I_l$ or $J_r$. Given a vector $\gamma \in \mathbb{R}^k$, we also introduce a \textit{weighted Hamming norm} on $\mathbb{R}^k$ as $\norm{x}_{\gamma} = \sum_{j=1}^k\gamma_j|x_j|$. We can now state the proof-by-contraction method, a sufficient condition for the validity of a given entropy inequality~\cite{Bao:2015bfa}}:
\begin{nthm}\label{thm:contraction}
    Let $f:\{0,1\}^L\to\{0,1\}^R$ be a $\norm{\cdot}_{\alpha}$-$\norm{\cdot}_{\beta}$ contraction, i.e.
    \begin{equation}
        \norm{x - x'}_{\alpha} \geq \norm{f(x) - f(x')}_{\beta} \qquad \forall x, x'\in\{0,1\}^L.
    \end{equation}
    If $f\left(x^{(i)}\right) = y^{(i)}$ for all $i=1,\ldots,n+1$, then \eqref{eq:candidate_ineq} is a valid entropy inequality on graphs.
\end{nthm}
\update{The final requirement on occurrence bitstrings is a simple consistency condition that only the appropriate external vertices from LHS min-cuts are included by $f$ in forming cuts for each RHS set $J_{r}$. If $f$ is a function satisfying the contraction condition in Theorem \ref{thm:contraction}, we call $f$ a \textit{contraction map}.} For the detailed proof, we refer the interested reader to~\cite{Bao:2015bfa,Avis:2021xnz}, but we briefly review the elements here. At the core of the proof is constructing a partitioning of the underlying graph into $2^L$ subsets of vertices as determined by inclusion/exclusion with respect to each of the $L$ min-cuts of LHS terms in \eqref{eq:candidate_ineq}. Each of these $2^L$ subsets is conveniently labeled by a bitstring $x\in\{0,1\}^L$. Given a bitstring $x$, one can reconstruct the associated vertex subset via $W(x) = \bigcap_{l=1}^L W_l^{x_l}$, where $W_l^1=W_l$ is the min-cut for the $l^{\text{th}}$ region, and $W_l^0 = {W_l}^\complement$ is its complement. This may also be inverted to recover the cuts themselves via $W_l = \bigcup_{x_l=1}W(x)$, where the union runs over all $x\in\{0,1\}^L$ subject to $x_l=1$.
While the entropy will depend on the actual weights of the edges in a graph, the construction of these subsets $W(x)$ allows one to keep track of how many times we count edge weights coming out of that subset when computing entropies.

The contraction map $f$ takes bitstrings of length $L$ to bitstrings of length $R$, and in this way utilizes the subsets $W(x)$ to build (not necessarily minimal) cuts $U_r$ for each of the RHS subsystems. More explicitly, these cuts are specified in terms of $f$ by $U_r = \bigcup_{f(x)_r=1}W(x)$ (cf. the reconstruction of the LHS cuts $W_l$ above). By minimality of cuts, the sum of edge weights involved in every $U_r$ will upper bound the true entropy of terms on the RHS. Thus, $f$ can be thought of as giving a sort of naive overestimate of the entropies of terms on the RHS, as long as the cuts $U_r$ are appropriately faithful (e.g., vertices in a cut for $I$ are not all excluded from vertices in a cut for $I\cup J$). 

The contraction condition on $f$, while looking cryptic and abstract since it is defined on bitstrings, enforces that edges coming out of the $W(x)$ are counted more times on the LHS of the inequality than on the RHS. Because it holds for every pair of bitstrings, which correspond to different sets $W(x)$, it says that edges between different $W(x)$ were counted no more times on the RHS than on the LHS, hence proving the inequality~\eqref{eq:candidate_ineq} valid regardless of the edge weights in the underlying graph. \update{Determining the existence of a contraction map for a given entropy inequality can be accomplished via greedy algorithms which are guaranteed to find one if it exists \cite{Bao:2015bfa}. Unfortunately though, this rapidly becomes computationally intractable, since one is effectively searching in a space of all possible $2^{R\times(2^L-(n+1))}$ maps $f$ compatible with the last line of Theorem~\ref{thm:contraction} as specified by their image bitstrings.\footnote{See \cite{Avis:2021xnz} for a more refined study of the space of possible contraction maps.}}

\subsection{Generalizing to hypergraphs}
In \cite{Bao:2020zgx}, this was generalized to hypergraph models, where, again, entropies were declared to be calculated by weights of min-cuts. The sets $W(x)$ were again useful, but because hyperedges can connect more than two vertices at a time, one had to consider not just pairwise interactions between the $W(x)$, but also general $k$-edges, i.e., hyperedges connecting $k$ vertices. Demanding that $k$-edges on the LHS were counted no fewer times than on the RHS, essentially imposed additional criteria on the contraction map $f$, which neatly nested into a sum over the rank $k$ of the involved hyperedges. This is captured in a generalization of the contraction condition:
\begin{nthm}\label{thm:gencontraction}
    Let $f:\{0,1\}^L\to\{0,1\}^R$ be an $i^k_{\alpha}$-$i^k_{\beta}$ contraction:
    \begin{equation}
        i^k_{\alpha}(x^1,\ldots,x^k) \geq i^k_{\beta}(f(x^1),\ldots,f(x^k)) \qquad \forall x^1,\ldots,x^k\in\{0,1\}^L.
    \end{equation}
    If $f\left(x^{(i)}\right) = y^{(i)}$ for all $i\in[n+1]$, then \eqref{eq:candidate_ineq} is a valid inequality on rank-$k$ hypergraphs.
\end{nthm}
Here, $i^k_\alpha(x^1,\ldots,x^k) = \sum_l\alpha_li^k(x^1_l,\ldots,x^k_l)$ and $i^k:\{0,1\}^k\to\{0,1\}$ is an indicator function that returns 0 if the input bitstring is all 0s or all 1s, and 1 otherwise. Because a valid contraction map $f$ for hypergraphs must obey additional constraints, an inequality that is true for graphs may no longer hold for hypergraphs. Thus, the hypergraph entropy cone has fewer constraints and is larger. That this containment is strict was shown in~\cite{Bao:2020zgx}.

\section{Link models}\label{sec:link models}
\subsection{Introduction}
Inspired by the previous models, we now introduce the link model.
In this paper, a \textit{link} consists of standard $1$-dimensional \textit{knots} in $3$-dimensional Euclidean space. For our purposes, knots will often be trivial (the \textit{unknot}) and just referred to as \textit{loops}. Letting $\mathcal{L}$ denote the \update{finite set} of loops, we distinguish a subset $\partial \mathcal{L} \subseteq \mathcal{L}$ as \textit{external} and call the others \textit{internal}. All loops are assigned some nonnegative weight. Notice that links no longer carry distinct vertex/edge-like components -- loops will implement both functionalities now. In analogy with previous discussions, a bijection $\partial \mathcal{L}\to[n+1]$ once again identifies the external loops as parties of a quantum state.

\update{To define the notion of loop cuts on a link model in analogy with edge cuts on graph models, we will need the following basic idea:}

\begin{ndefi}
    \update{An $n$-loop link $L$ embedded in $\mathbb{S}^3$ is said to be \emph{topologically linked} if for any choice of embedded disks $\{D_i\subset\mathbb{S}^3\}_{i=1}^n$ such that $L=\bigcup_{i=1}^n \partial D_i$, the set $\;\bigcup_{i=1}^n D_i$ is topologically connected. Otherwise $L$ is said to be \emph{topologically unlinked}.}
\end{ndefi}

\update{Given a link with loop set $\mathcal{L}$, any subset of loops $\tilde{\mathcal{L}}\subseteq \mathcal{L}$ defines a \emph{sublink}. Similar to the notion of connected components in a topological space, one easily sees that a general link $L$ may always be decomposed into \emph{connected sublinks}, namely, maximal sublinks which are topologically linked. Hence we can equivalently understand a link $L$ as being topologically unlinked if and only if it consists of more than one connected sublink. Clearly, topological linking is a property which can be lost when some subset of loops is removed from a link, in the sense that the remaining loops can break into multiple connected sublinks. Following this intuition, we will say that two loops of a link $L$ are topologically unlinked in $L$ if they belong to distinct connected sublinks, and similarly that two subsets of loops are topologically unlinked in $L$ if every pair of respective loops belong to distinct connected sublinks.}\footnote{\label{fn:sublooplinks}\update{The converses are clearly as follows: two loops of a link $L$ are topologically linked in $L$ if they belong to the same connected sublink, and similarly two subsets of loops are topologically linked in $L$ if all their loops belong to the same connected sublink. However, let us emphasize an important subtlety: two loops of a link $L$ which are topologically linked in $L$ need not be so if taken to define a link on their own. For instance, if $L$ is a Brunnian link, any two loops in it are topologically linked in $L$, but if taken to define a link on their on (cf. deleting the other loop), they would be topologically unlinked (indeed, they would define a $2$-component unlink). We will come back to this in Section \ref{ssec:toolkit}.}}

\update{We use the above to define a \emph{loop cut} $C$ for a subsystem $I\subseteq[n]$ as any set of internal loops $C\subseteq \mathcal{L}\smallsetminus\partial \mathcal{L}$ which, if removed, give a link where the loops in the subsystem $I$ are topologically unlinked from all other external loops in $\partial \mathcal{L}$. The weight of $C$ is given by the sum of its loop weights, and denoted by $\abs{C}$.} The entropy of $I$ is defined as the minimum weight of all possible loop cuts $C$ for $I$. Any such \emph{loop min-cut} for $I$ will be denoted by $C_I$.\footnote{When there exist multiple min-cuts with the same cut weight, one simply makes a choice and sticks to it.} In the spirit of previous work, we posit that the entropy of a subsystem $I$ is given by the weight of a loop min-cut $C_I$. Note that both the graph and hypergraph models are easily reproducible in the link model by converting vertices to infinite-weight loops and $k$-edges to finite weight loops, all Hopf-linked to their vertex loops, as illustrated in Figure~\ref{fig:link_conversion}.

\begin{figure}
    \centering
    \includegraphics[width=.8\textwidth]{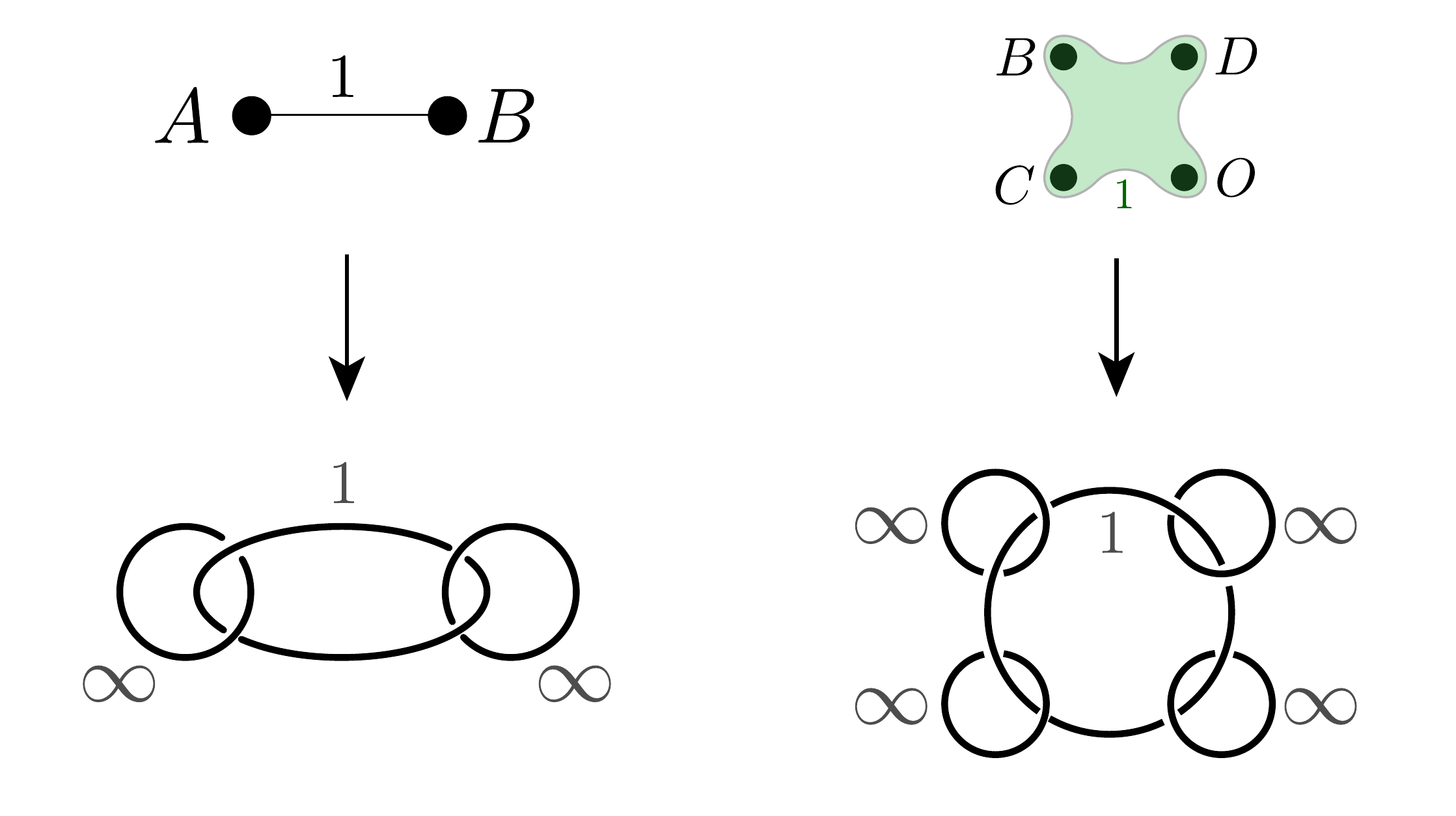}
    \caption{Link construction from graph and hypergraph models. One can reproduce all min-cut weights appropriately by making the following substitutions. Vertices do not contribute to the cut weight, so they should be replaced with loops of infinite weight. Loops of finite weight can be used to replace both edges and hyperedges where the vertices are Hopf-linked to the edge loop. Note that link models inherently no longer distinguish between vertices and edges.}
    \label{fig:link_conversion}
\end{figure}

\subsection{Capturing the entropies of ray 15}
To motivate the construction of the link model, we share an explicit quantum state whose subsystem entropies can be represented by the link model yet provably fail to be realizable using hypergraph models~\cite{Bao:2020zgx,Bao:2020mqq}. The pertinent entropy vector, henceforth referred to as ray 15 following~\cite{Bao:2020mqq}, comes from a particular 5-party mixed graph state,\footnote{Note that the graph state described here is in the context of~\cite{Hein_2004}, and not a state whose entropy vector can be realized by graph representation.} where the vertices denote spin systems and edges represent Ising-type interactions~\cite{Hein_2004}. Explicitly, ray 15 reads
\begin{equation}\label{eq:gs15}
    \Vec{S}_{15} = (1\,1\,1\,1\,1\,;1\,2\,2\,2\,2\,2\,2\,2\,2\,1\,;2\,2\,2\,2\,2\,2\,2\,2\,2\,2\,;2\,2\,1\,2\,2\,;1),
\end{equation}
where entropies are listed in lexicographical order, with subsystems of increasing number of parties separated by a semicolon. Notably, all graph states are stabilizer states, and hence, ray 15 must lie within the 5-party stabilizer entropy cone. 

\begin{figure}
    \centering
    \includegraphics[width=.8\textwidth]{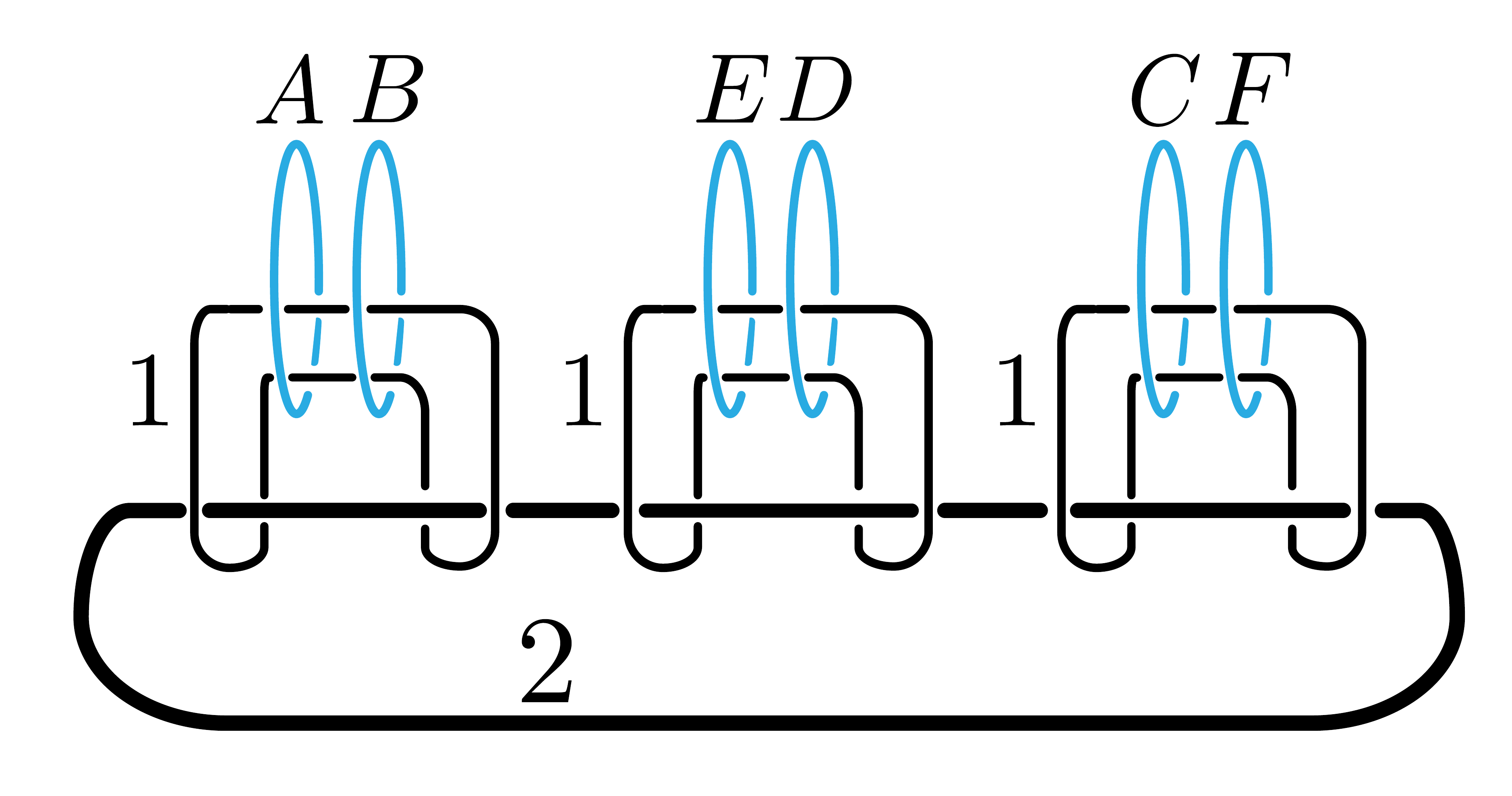}
    \caption{The topological link whose entropies reproduce ray 15 from~\cite{Bao:2020mqq}. $A$, $B$, $C$, $D$, and $E$ are the parties in question, with $F$ taken to be the purifier. The min-cut of any single external loop has weight $1$, as do the pairs $AB$, $ED$, and $CF$. All other pairs of external loops have min-cut weight $2$, and so do all triples. The min-cut weight of any other subsystem can be obtained from that of its complement with respect to all external loops.}
    \label{fig:link_15}
\end{figure}

In attempting to determine whether the hypergraph and stabilizer entropy cones coincide, the authors found that the answer is negative~\cite{Bao:2020mqq}, indicating the strict containment of the hypergraph entropy cone within the stabilizer entropy cone~\cite{Walter:2020zvt}. Indeed, ray 15 violates a linear entropy inequality that all hypergraph models were shown to obey in~\cite{Bao:2020mqq}:
\begin{equation}\label{eq:ineq15}
    S_{AB} + S_{DE} + S_{ACD} + 2 S_{ACE} + S_{BCD} + S_{ABDE} \geq S_{AC} + S_{AE} + S_{BD} + 2 S_{ABCD} + S_{ACDE}.
\end{equation}
In Figure \ref{fig:link_15}, we give an explicit link model that reproduces ray 15, thereby demonstrating that there exist link models which reproduce entropy vectors inaccessible by hypergraphs alone. Moreover, because all hypergraph models can be reproduced in the link model, the entropy cone of links properly contains the hypergraph entropy cone.

This construction motivates us to investigate the degree to which link models can represent multipartite entanglement of generic quantum states. We will consider whether the contraction-map technology can be generalized to link models, allowing us to extend the entropy cone of hypergraphs to that of link models. This will pave the way for exploring whether the link model entropy cone can possibly contain the stabilizer entropy cone of~\cite{Linden:2013kal} or even the full quantum entropy cone~\cite{pippenger2003inequalities}. \update{One can consider a program of using contraction maps to directly prove candidate entropy inequalities on link models, and exploring such inequalities as candidate entropy inequalities for the currently unknown higher-party stabilizer and quantum entropy cones (e.g. by testing whether they are violated by the entropies of known stabilizer or quantum states). In this way, the contraction map methods and associated graph/link models serve as an immediate application and useful set of tools in the study of quantum state entropies and their constraints. We further discuss the application of link models to classifying the structure of entanglement of quantum states in Section \ref{sec:disc}.}

\subsection{The link model toolkit}
\label{ssec:toolkit}

Studying general link models will require the introduction of various definitions and constructs, to which we now turn. Given some link, we say that a sublink $\tilde{\mathcal{L}}$ of it is an \emph{irreducible sublink} if it is intrinsically topologically linked, i.e. if it is topologically linked once all loops in $\mathcal{L}\smallsetminus\tilde{\mathcal{L}}$ are deleted.\footnote{Note that a sublink of some link $L$ may not be irreducible even if its constituent loops happen to be topologically linked in $L$. -- see footnote \ref{fn:sublooplinks} for an example.} A simple class of irreducible sublinks are the connected sublinks into which any link may be decomposed.

Any loop cut $C$ for $I$ naturally defines two additional subsets of loops. Let $W_C \subseteq \mathcal{L} \smallsetminus C$ be the minimal-cardinality subset of loops with $W_C \cap \partial \mathcal{L} = I$ and such that, once the loops in $C$ are removed, form a collection of connected \update{sublinks}. 
In practice, $W_C$ can be found by collecting all loops in the link defined by $\mathcal{L}\smallsetminus C$ that are topologically \update{linked} to the pertinent external loops $I$. The resulting set $W_C$ is clearly unique and referred to as the \textit{cut interior} of $C$. Similarly, $\overline{W}_C = \mathcal{L} \smallsetminus (W_C\cup C)$ is also unique and referred to as the \textit{cut exterior} of $C$. The cut exterior contains all external loops $\partial \mathcal{L} \smallsetminus I$ which are absent from $W_C$, and consists of connected \update{sublinks} as well. When talking about a loop min-cut $C_I$, we will denote its associated cut interior as $W_I$ for simplicity. As an example, consider looking at subsystem $I=AC$ in Figure \ref{fig:link_15}: a valid loop min-cut $C_{AC}$ for it is the weight-$2$ loop which, when removed, defines the cut interior $W_{AC}=\{A,C\}$ alone -- the corresponding cut exterior $\overline{W}_{AC}$ consists of the three weight-$1$ loops and the complementary external loops $B$, $D$, $E$ and $F$.

To refine our analysis of loop cuts, we introduce a notion of \textit{prunability}. For a given subsystem $I$, the loop min-cut $C_{I}$ by definition contains loops which topologically unlink $W_{I}$ from $\overline{W}_{I}$. Consider a sublink $B_I\subseteq \mathcal{L}$ containing $k$ loops that is irreducible and has nontrivial intersection with all of $W_{I}, \overline{W}_{I},$ and $C_{I}$. Let us refer to such a sublink as a \textit{bridge} across the loop min-cut $C_I$. By definition, bridges clearly have cardinality $k\ge3$. We call a bridge \textit{prunable} if it remains a bridge upon removal of any nontrivial subset of its loops. We say a bridge is \textit{minimal} if it cannot be pruned. A minimal bridge of loop cardinality $k$ is called a \textit{$k$-bridge} and denoted by $B^k_I$. These turn out to be crucial objects for link models of entanglement to not trivialize down to hypergraphs. Within this framework, hypergraphs consists only of Hopf-linked $3$-bridges. One can explicitly construct a link model from any hypergraph model by replacing vertices with infinite weight loops and (hyper)edges with loops of the corresponding weight -- see Figure~\ref{fig:link_conversion} for an illustration. 
The richer entanglement structure that link models capture is due to $3$-bridges with different topologies as well as the existence of $k$-bridges with $k\geq 4$. As an example of a topologically interesting $3$-bridge, consider the link model in Figure~\ref{fig:link_15}. If we take $I=BCD$, the loop min-cut consists of the weight-$2$ loop. An associated $3$-bridge consists of $B$ in the cut interior, the weight-$1$ loop in the cut exterior, and the weight-$2$ loop in the loop min-cut. This 3-bridge does not involve Hopf-linked loops -- instead, as a sublink, it is a 3-component Brunnian link.

Loops in $C_{I}$ can be classified by the minimal cardinality of the $k$-bridges for $I$ that they contribute to. In particular, we say that $\ell\in C_{I}$ is a \textit{$k$-loop} for $I$ if, among all minimal bridges $B\ni\ell$ for $I$, the smallest bridge cardinality is $k$. For instance, one can easily verify that the weight-$2$ loop in Figure~\ref{fig:link_15} is a 3-loop for $I=BCD$.
We can stratify the loops $\ell\in C_{I}$ into subsets $C^{(k)}_{I}\subseteq C_{I}$ such that $\ell$ is a $k$-loop for $I$ if and only if $\ell\in C^{(k)}_{I}$. Importantly, the subsets $C^{(k)}_{I}$ are disjoint for different $k$ and define a finite spanning partition of $C_{I}$ for any link with finite $\abs{\mathcal{L}}$, i.e., $\bigcup_{k=3}^{\abs{\mathcal{L}}} C^{(k)}_{I} = C_{I}$. When computing the total weight of a min-cut $C_{I}$, we thus have 
\begin{equation}
\label{eq:bridgedecom}
    \abs{C_{I}} = \sum_{k=3}^{\abs{\mathcal{L}}} \abs{C^{(k)}_{I}} = \sum_{k=3}^{\abs{\mathcal{L}}} \sum_{\ell\in C^{(k)}_{I}} w(\ell),
\end{equation}
where $w(\ell)$ is the weight of $\ell$.

\subsection{A topological min-cut bottleneck}
In previous work, the simple connectivity structure of graphs and hypergraphs implied that a cut, given by a subset of vertices, automatically specified the (hyper)edges whose weights contributed to the cut weight. When it comes to link models, however, loops subsume both the role of edges and vertices. While powerful for capturing a wider class of entropy vectors, this generality comes with a price. The analogous cut interior $W_C$ now no longer specifies uniquely the loop cut $C$ that gave rise to it and whose weight would determine the cut weight.

In order to address this new issue, one must consider the topological link classification problem~\cite{habegger1990classification}; that is, given a topological link, determining whether a sublink is topologically linked to the remainder of the link, or whether they are, in fact, unlinked. This problem is an outstanding problem in topology which we do not propose a solution to. Instead, we will build our formalism on top of a posited solution. We parameterize our ignorance in the form of an oracle that essentially does the difficult work of obtaining the min-cut structure of a general link.

Although such an oracle already performs a large task, we attempt to pare its power to the minimal necessary for our purposes. Specifically, consider an oracle called \texttt{Link-Min-Cut} that takes as input a link $\mathcal{L}$ and a set of external loops $I \subseteq \partial\mathcal{L}$, and outputs its corresponding loop min-cut $C_I$. One can then construct the cut interior $W_I$ as described above by taking the set of external loops $I$ and finding the smallest number of connected \update{sublinks} that contain them after $C_I$ is removed.

Building on the \texttt{Link-Min-Cut} routine, we employ the following additional subroutine:
\begin{equation}
    \texttt{Bridge-Oracle}(\mathcal{L}, I, B)=
    \begin{cases}
        1 & \quad \text{if $B$ is a $\abs{B}$-bridge across $C_I$},\\
        0 & \quad\text{otherwise.}
    \end{cases}
\end{equation}
where $B$ is a set of loops, a candidate minimal bridge for $I$, and the required knowledge of $C_I$ is obtained from the output of \texttt{Link-Min-Cut} for the same $\mathcal{L}$ and $I$ arguments. The ultimate purpose of \texttt{Bridge-Oracle} is to solve a decision problem, namely, whether $B$ serves as a $|B|$-bridge for the set of external loops in $I$. In the generalization of the contraction-map proof method, this oracle will be used for determining which loops contribute to the entropy of the link model.

\subsection{Generalization of contraction maps}
The contraction-map method to prove entropy inequalities involves a clever partitioning of the vertex set in (hyper)graph models. For a given inequality as in \eqref{eq:candidate_ineq} with indices $\{I_l\subseteq[n]\}_{l\in[L]}$ on its LHS, every vertex could be labeled by a bitstring of length $L$ where each bit encoded whether it was contained in the min-cut for each $I_l$. For the link model, we need to generalize this addressing to loops in terms of loop min-cuts, which behave differently. Our goal in this subsection will be to prove that such a generalization exists, given access to an appropriate oracle for bridges.

We first introduce some notation and useful constructions. Pick a loop min-cut $C_{l}$ for each index $I_l$ on the LHS of the entropy inequality, and let $W_l$ and $\overline{W}_l$ be respectively its cut interior and exterior. To each loop $\ell\in \mathcal{L}$, we assign an upgraded bitstring $x\in\{-1,0,1\}^L$ via
\begin{equation}
    x_{l} = \delta(\ell\in W_{l}) - \delta(\ell\in \overline{W}_{l}).
\end{equation}
In other words, $x_l$ is set to $-1$, $0$ or $1$ if $\ell$ is in $\overline{W}_l$, $C_l$, or $W_l$, respectively.
We can now partition the loop set $\mathcal{L}$ into disjoint subsets of loops via
\begin{equation}
\label{eq:bitsubw}
    W(x) = \bigcap_{l\in[L]} W_{l}^{x_l},
\end{equation}
where we are declaring $W_l^1\equiv W_l$, $W_l^0 \equiv  C_l$ and $W_l^{-1}\equiv \overline{W}_l$. There will be $3^L$ such sets of loops (i.e., as many as bitstrings), with some possibly empty. Note that we can reconstruct the cut interior $W_l$ of any $I_l$ in terms of the $W(x)$ subsets via
\begin{equation}
\label{eq:recwl}
    W_{l} = \bigcup_{x:x_l = 1} W(x),
\end{equation}
and similarly for the loop min-cuts using
\begin{equation}
\label{eq:reccl}
    C_{l} = \bigcup_{x:x_l = 0} W(x).
\end{equation}
In (hyper)graph models, one was easily able to determine the edges contributing to a min-cut (cf. $C_{l}$) by just knowing the vertex min-cut (cf. $W_l$). In contrast, in the link model the cut interior $W_l$ does not immediately determine what the loop min-cut $C_l$ is in any obvious way for a general link.\footnote{Recall that it was $C_l$ which defined $W_l$, and not the converse.} This explains why we needed an upgraded bitstring addressing encoding these two objects separately. In fact, we expect an explicit specification of such a map $W_l\mapsto C_l$ to be a highly nontrivial problem in knot theory which we will not try to solve completely. Instead, we will try to break down this map into pieces suitable for a proof by contraction. However, with these objects defined, we can now give the formal statement of the generalized contraction map:
\begin{nprop}\label{prop:generalized contraction}
    An entropy inequality in the form of \eqref{eq:candidate_ineq} is valid on link models if:
    \begin{enumerate}
        \item There exists a map $f:\{-1,0,1\}^L \to \{-1,0,1\}^R$ from min-cuts for regions on the LHS of \eqref{eq:candidate_ineq} to valid cuts for regions on the RHS of \eqref{eq:candidate_ineq}; that is, the regions defined by
        \begin{equation}
        \label{eq:rhscuts}
            U_r = \bigcup_{x:f(x)_r=1}W(x)
        \end{equation}
        are (not necessarily minimal) cuts for $J_r$.
        \item For every $3\le n \le k\le |\mathcal{L}|$, $f$ obeys the cut-dependent contraction property
        \begin{equation}
    \sum_{l=1}^L\alpha_l F^{(n,k)}_l(x^1,\ldots,x^n) \geq \sum_{r=1}^R\beta_r \delta(f(x^1)_r = 0) \sum_{l=1}^{L} F^{(n,k)}_l(x^1,\ldots,x^n),
    \end{equation}
    where $F_{l}^{(n,k)}:\{-1,0,1\}^{L\times n}\to\{0,1\}$ is an oracular indicator function that returns 1 if the input $n$-tuple of bitstrings corresponds to a $k$-bridge.\footnote{Double-counting of loop contributions is avoided by constructing $F_{l}^{(n,k)}$ as in Algorithm \ref{algo}.}
    \end{enumerate}
\end{nprop}

\update{We will need a few intermediate ingredients to prove this result, as we describe next.} To begin, we will think of the loop cut $C_l$ as a collection of loops independently satisfying some property with respect to $W_l$. Paralleling \eqref{eq:bridgedecom}, we begin by decomposing $C_l$ into sets of $k$-loops through
\begin{equation}
    C_l = \bigcup_{k=3}^{\abs{\mathcal{L}}} C_l^{(k)}.
\end{equation}
By definition, a loop $\ell \in C_l$ is a $k$-loop $\ell \in C_l^{(k)}$ for the smallest $k$ such that it belongs to a $k$-bridge $B_l^k$, where $B_l^k$ by construction intersects all of $W_l$, $\overline{W}_l$, and $C_l$ nontrivially. In particular, notice that if $\ell \in C_l$, then it must be that $\ell\in W(x)$ for some bitstring with $x_l=0$. \update{However, this decomposition is not fine-grained enough for our purposes: for a contraction condition to hold, we need to be able to sum over bitstrings, and the decomposition above does not contain enough information. Recall \eqref{eq:bitsubw}, which defines $W(x)$ for a bitstring $x$ in terms of an intersections of vertex subsets $W_l$ or their complements. However, a given $k$-bridge will generally intersect multiple $W(x)$, and the decomposition into $k$-loops does not let us determine which bridge $B_l^k$ or set $W(x)$ it belongs to. That is, we are unable to pass directly from $k$-loops to a condition on bitstrings, because we cannot identify bitstrings with specific $k$-loops that contribute to a cut.}

To handle this, we will introduce a further decomposition of a $k$-loop. Firstly, assuming access to an appropriate oracle for validating the existence of a bridge, we show that it is possible to decompose a cut $C_l$ such that each loop in $C_l$ belongs to a distinct bridge:
\update{
\begin{nlemma}\label{lem:single loop}
    Consider a loop min-cut $C_l$ and $k$-bridge $B_l^k$ for the boundary region $I_l$. Then $|B_l^k \cap C_l| = 1$.
\end{nlemma}
}
\begin{proof}
\update{
We prove this by contradiction. Suppose $|B_l^k \cap C_l| > 1$, and let $\ell_i \in B_l^k \cap C_l$. By definition of bridge minimality, $B_l^k$ is not prunable, so removing $\ell_i$ will cause $B_l^k$ to no longer be a bridge. This means that $B_l^k\setminus \{\ell_i\}$ is either no longer irreducible, does not have nontrivial intersection with $W_l$, $\overline{W}_l$, and $C_l$, or both. However, because $\ell_i \in C_l$ and there exists at least one $\ell_j\in C_l$ with $j\neq i$, the last two options are not possible. Hence, it must be the case that $B_l^k\setminus \{\ell_i\}$ is now reducible, meaning that the sublink defined by the remaining loops contains more than one connected sublink.
}

\update{
If one of the resulting connected sublinks still has nontrivial intersection with $W_l$, $\overline{W}_l$, and $C_l$, then the resulting loop set is a $k'$ bridge with $k' < k$, meaning $B_l^k$ was prunable and not minimal, yielding a contradiction. If, instead, none of the resulting connected sublinks has non-trivial intersection with $W_l$, $\overline{W}_l$, and $C_l$, then at least one of $(B_l^k\setminus \{\ell_i\}) \cap W_l$, $(B_l^k\setminus \{\ell_i\}) \cap \overline{W}_l$, or $(B_l^k\setminus \{\ell_i\}) \cap C_l$ is topologically unlinked from the other two loop sets. If $(B_l^k\setminus \{\ell_i\}) \cap C_l$ becomes topologically unlinked from the other two loop sets, then it must be the case that, in fact, all three loop sets are topologically unlinked from each other, for otherwise $C_l$ would have not defined a cut for the pertinent subsystem in the first place. Hence, it must be the case that $(B_l^k\setminus \{\ell_i\}) \cap W_l$ and $(B_l^k\setminus \{\ell_i\}) \cap \overline{W}_l$ are topologically unlinked, regardless of what happens to $(B_l^k\setminus \{\ell_i\}) \cap C_l$.
However, this implies $C_l$ was not minimal in weight, as we could instead define an alternative cut $C_l' = C_l\setminus\{\ell_j \in C_l : j \neq i\}$ for the same subsystem, with $(B_l^k\setminus \{\ell_i\}) \cap C_l$ now being part of either the new interior or exterior. As noted above, removing $C_l'$ topologically unlinks the new interior and exterior, so it is a valid cut. We also see that removing $C'_l$ makes $B_l^k$ topologically unlinked, and the weight of $C'_l$ is strictly smaller than that of $C_l$, because $C'_l \subset C_l$. We conclude that we must have $|B_l^k \cap C_l| = 1$.
}
\end{proof}

Lemma \ref{lem:single loop} will allow us to reduce our analysis of a cut $C_l$ to an analysis of the distinct bridges to make up each of its loops. Applying the lemma, we can separate the loops within a $k$-bridge $B_l^k$ into the loop in $C_l$, and the remaining $k-1$ loops in $B_l^k$ that belong to either $W_l$ or $\overline{W}_l$. Those loops in $B_l^k\cap W_l$ ($B_l^k\cap \overline{W}_l$) will necessarily belong to some $W(x)$ sets with $x_l=1$ ($x_l=-1$). However, notice that these additional $k-1$ loops need not all belong to distinct $W(x)$ sets -- instead, in general there will be between $2$ and $k-1$ such sets containing them, as mentioned previously. We now introduce the further decomposition of loop-cuts $C_l^{(n,k)}\subseteq C_l^{(k)}$ consisting of those $k$-loops in $C_l$ for which the smallest number of $W(x)$ sets needed to cover any one of its associated $k$-bridges is $n$. Enumerating bitstrings by a superscript as $x^i$, this means that every $\ell\in C_l^{(n,k)}$ is a $k$-loop with $\bigcup_{i=1}^{n} W(x^i)$ containing $B_k$ minimally in the sense that $n$ is the smallest across all associated $k$-bridges. Hence:
\update{
\begin{ncor}\label{cor:clnk}
Every $C_l^{(k)}$ can be further decomposed as
\begin{equation}
    C_l^{(k)} = \bigcup_{n=3}^kC_l^{(n,k)}
\end{equation}
where
\begin{equation}
\label{eq:clnk}
    C_l^{(n,k)} = \left\{ \ell\in W(x^{1}) \,:\, x^{1}_l=0, \;
    \bigcup_{i=1}^n W(x^i) \supseteq B_k \ni \ell \text{ minimally}
    \right\}_{(x^{1},\ldots,x^{n})},
\end{equation}
where the subscript $(x^{1},\ldots,x^{n})$ implies iteration over all $n$-tuples of bitstrings in $\{-1,0,1\}^L$, and the choice $x^{1}_l=0$ is without loss of generality. 
\end{ncor}
}

Notice that $C_l^{(n,k)}$ is nontrivial only if $n,k\ge3$, for otherwise the conditions that define it are impossible to satisfy. At the level of bitstrings, we can see that the condition in \eqref{eq:clnk} that $\bigcup_{i=1}^n W(x^i)$ contain a $k$-bridge minimally necessarily means that $\sum_{i=1}^n |x^{i}_l| = n-1$ (i.e., only $x^{1}_l=0$), and that $\abs{\sum_{i=1}^n x^{i}_l} \le n-3$ (i.e., there are both $+1$ and $-1$ bits). While necessary, these bitstrings conditions are not sufficient to guarantee those in \eqref{eq:clnk} -- this issue will reappear when formulating the contraction map. In the contraction-map method, we want to turn a sum over the $k$-tuples of bitstrings satisfying the conditions in \eqref{eq:clnk} into a sum over all $k$-tuples of bitstrings, with some indicator function specifying which tuples do participate in each $C_l^{(n,k)}$; however, we cannot formulate the necessary and sufficient conditions for a $k$-bridge purely in terms of elementary properties of a bitstring. Hence, we will need access to an oracle:
\update{
\begin{ndefi}[Oracular Indicator]
Consider a candidate entropy inequality in the canonical form of \eqref{eq:candidate_ineq} on an $n$-party link model. The \emph{oracular indicator} is a Boolean map on bitstrings, $F^{(n,k)}_l:\{-1,0,1\}^{nL} \to \{0,1\}$, which returns 1 when all of the necessary and sufficient conditions for a $k$-bridge to be minimally contained in $\bigcup_{i=1}^n W(x^i)$ are satisfied, and 0 otherwise.
\end{ndefi}
}
Note that, following the convention of \eqref{eq:clnk}, we may without loss of generality choose to have $F_l^{(n,k)}(x^1,\ldots,x^n) = 0$ when $x^1_{l} \neq 0$. The function also has one additional feature beyond checking the $k$-bridge conditions: it keeps track of what loops have \textit{already} \update{contributed to the cut weight} and may be part of multiple $k$-bridges.\footnote{In the context of the link given in Figure~\ref{fig:link_15}, notice that the weight-2 loop contributes to forming 3- but also 4-bridges, but because of the coloring strategy described below it is only counted once as a 3-loop.} This can be implemented by associating each loop in the cut set with the $k$-bridge of smallest $k$ in which it participates, with an additional ordering for $k$-bridges of the same valence. An example for how such a function could be implemented, assuming access to \texttt{Link-Min-Cut} and \texttt{Bridge-Oracle} is:

\vspace{10pt}
\begin{algorithm}[H]
\label{algo}
\caption{Implementation of $F^{(n,k)}_l$}\vspace{5pt}
\SetAlgoLined
Calculate min-cut $C_l$ for subsystem $I_l$ via \texttt{Link-Min-Cut};\\
Color all loops in $C_l$ gray and set all values of $F^{(n,k)}_l(x^1,\ldots,x^n)$ to zero;\\
$k \coloneqq 3$;\\
\While{$C_l$ contains gray loops}{
 $n \coloneqq 3$;\\
\While{$n \leq k$}{
  \For{each $n$-tuple of bitstrings $(x^1,\ldots,x^n)$ with $x_l^1 = 0$}{\vspace{1pt}
    $\mathbf{W} \coloneqq \bigcup_{i=1}^n W(x^i)$;\\
      \For{each subset of $k$ loops $B \subseteq \mathbf{W}$}{\vspace{1pt}
          \If{{\upshape\texttt{Bridge-Oracle}$(\mathcal{L},I_l,B)$}}{\vspace{1pt}
            $\ell \coloneqq \text{loop in } B \bigcap W(x^1)$;\\
            \If{$\ell$ exists and is gray }{\vspace{1pt}
                Color $\ell$ green;\\
                Set $F^{(n,k)}_l(x^1,\ldots,x^n)=1$;
            }
           }
      }
  }
  $n \coloneqq n+1$;
  }
  $k \coloneqq k+1$;
 }
\end{algorithm}
\vspace{10pt}

To be explicit, the oracular nature of $F_l^{(n,k)}$ manifests purely in determining the sufficient conditions for $B$ to be a $k$-bridge -- the other features are for presentation/convenience, and can be implemented as separate functions on bitstrings.

Finally, we comment on the first condition appearing in the statement of Proposition \ref{prop:generalized contraction}, which is an extra requirement on any possible contraction map $f$, in addition to the contraction property. Recall that for a cut $C_l$ for a region $I_l$ appearing on the LHS of a candidate entropy inequality \eqref{eq:candidate_ineq}, we were able to decompose the cut as in \eqref{eq:reccl}.
In order to relate the weights of these cuts to those on the RHS, we need to perform a similar decomposition of the RHS cuts. However, the caveat is that the decomposition must still be in terms of the $W(x)$ (or else the comparison would be too difficult), and the $W(x)$ are defined only in terms of the regions appearing on the LHS of \eqref{eq:recwl}. That is, given a candidate contraction map $f:\{-1,0,1\}^L \to \{-1,0,1\}^R$, when we construct the corresponding vertex subsets for the RHS as in \eqref{eq:rhscuts}, each such $U_r$ must be a valid cut for the $r^{\text{th}}$ region appearing on the RHS. This means $U_r \cap \partial\mathcal{L} = J_r$ and $U_r$ is a cut interior for some (not necessarily minimal) loop cut $C_r$ for $J_r$. If such an initial condition cannot be satisfied, then the contraction map method has simply failed. That such a function exists is immediate in the hypergraph case, due to the bijective relation between cut vertices and edges. However, it is not clear that the same is true here. We interpret this fact as simply a manifestation of another stricter condition that is required for the contraction-map method to succeed.

\update{
With these tools in hand, we now provide a proof of Proposition \ref{prop:generalized contraction}:
}

\begin{proof}[Proof of Proposition \ref{prop:generalized contraction}]
\update{
Consider a candidate $n$-party entropy inequality in the canonical form of \eqref{eq:candidate_ineq}. Let $C_l$ be the minimal loop cut for the $l^{\text{th}}$ region appearing on the LHS. Then using the decomposition of $C_l$ into $C_l^{(n,k)}$ from Corollary \ref{cor:clnk} and the oracular indicator $F_l^{(n,k)}$, we can write the LHS as a sum over all bitstrings:
\begin{equation}\label{eq:LHS_cut_eqn}
    \sum_{l=1}^L\alpha_l|C_l| = \sum_{l=1}^L\alpha_l\sum_{k=3}^{|\mathcal{L}|}\sum_{n=3}^{k}|C_l^{(n,k)}| = \sum_{l=1}^L\alpha_l\sum_{k=3}^{|\mathcal{L}|}\sum_{(x^1,\ldots,x^n)}^{3\le n \le k}F_l^{(n,k)}(x^1,\ldots,x^n)|W(x^1)|.
\end{equation}
We continue by considering the RHS of \eqref{eq:candidate_ineq}. By the first assumption in the statement of the proposition, there exists a function $f:\{-1,0,1\}^L \to \{-1,0,1\}^R$ such that
\begin{equation}
\label{eq:rhscuts2}
    U_r = \bigcup_{x:f(x)_r=1}W(x),
\end{equation}
is a valid cut for the $r^{\text{th}}$ region appearing on the RHS. Notice that generally there are many implicit conditions $f$ will have to satisfy. In particular, the zeros of $f$ already specify RHS loop cuts, which themselves uniquely specify all RHS cut interiors and exteriors. This implies that all the $\pm 1$ entries of $f$ cannot be chosen freely, but are fixed by the zeroes of $f$. The latter also determine the weights of the RHS cuts $U_r$ given by \eqref{eq:rhscuts}, which can be expressed as follows:
\begin{equation}
   |U_r| = \sum_{k=3}^{|\mathcal{L}|}\sum_{(x^1,\ldots,x^n)}^{3\le n \le k}\sum_{l=1}^{L} \delta(f(x^1)_r = 0) F^{(n,k)}_l(x^1,\ldots,x^n)|W(x^1)|.
\end{equation}
This expression simply counts when loops from the LHS cuts end up contributing to the RHS.
Compared to \eqref{eq:LHS_cut_eqn}, which counts loop weights in the minimum cut for each of the LHS subsystems $I_l$, this expression counts the weight of those loops that contribute to a cut for the RHS subsystem $J_r$ via the assignment specified by $f$. Recall that the $r^{\text{th}}$ bit of $f(x^1)$ equaling $0$ means that $W(x^1)$ is being assigned to be part of the loop cut of $J_r$. Altogether, for the RHS we can write
\begin{equation}
    \sum_{r=1}^R\beta_r|U_r| = 
    \sum_{r=1}^R\beta_r \sum_{k=3}^{|\mathcal{L}|}\sum_{(x^1,\ldots,x^n)}^{3\le n \le k} \sum_{l=1}^{L} \delta(f(x^1)_r = 0) F^{(n,k)}_l(x^1,\ldots,x^n)|W(x^1)|.
\end{equation}
Assuming the cut specified by $f$ as in \eqref{eq:rhscuts} is indeed valid, then this will, by definition, upper bound the true value of the min-cut associated with the RHS. Hence, \eqref{eq:candidate_ineq} would then be satisfied if
\begin{equation}
\begin{aligned}
\sum_{k=3}^{|\mathcal{L}|} \sum_{(x^1,\ldots,x^n)}^{3\le n \le k} &\left[ \sum_{l=1}^L\alpha_l F^{(n,k)}_l(x^1,\ldots,x^n) \right] |W(x^1)| \\
&\geq \sum_{k=3}^{|\mathcal{L}|}\sum_{(x^1,\ldots,x^n)}^{3\le n \le k} \left[ \sum_{r=1}^R\beta_r \delta(f(x^1)_r = 0) \sum_{l=1}^{L} F^{(n,k)}_l(x^1,\ldots,x^n) \right] |W(x^1)|.
\end{aligned}
\end{equation}
A simple sufficient condition for this to hold is to demand that the inequality be satisfied for every $n$-tuple of bitstrings with $3\leq n\leq k$ for each value of $k$. Were this to be the case, then all possible equations of the following form would hold:
\begin{equation}
\label{eq:cp}
    \sum_{l=1}^L\alpha_l F^{(n,k)}_l(x^1,\ldots,x^n) \geq \sum_{r=1}^R\beta_r \delta(f(x^1)_r = 0) \sum_{l=1}^{L} F^{(n,k)}_l(x^1,\ldots,x^n).
\end{equation}
Note that this is the second condition given in the statement of Proposition \ref{prop:generalized contraction}. If the above is true, then the given entropy inequality is also true,
by the minimality property of min-cuts.\footnote{Namely, since by construction for the LHS we have min-cuts with $S(I_l)=|C_l|$, whereas for the RHS we have simply constructed cuts which are not necessarily minimal and thus obey $S(J_r)\le|C_r|$.}
}
\end{proof}

We remark that we obtain a contraction condition which is \textit{cut-dependent}, in the sense that it depends on the oracular indicator $F_l^{(n,k)}$. This cut-dependence of the contraction map is a new feature of link models that was not present in (hyper)graphs, arising from the new topological character of min-cuts.

\update{\subsection{Subadditivity and Strong Subadditivity}\label{subsec:ssa}
Although we have a generalized contraction condition in hand, an important consistency check that link models are relevant for quantum mechanics is whether the two universal quantum inequalities of subadditivity and strong subadditivity hold for topological links models. Indeed, it is known that these two inequalities completely determine the entropy cone of all quantum states for $n \leq 3$ parties.}

\update{It turns out that these properties do indeed easily follow by the definition of min-cuts on links, without requiring the contraction-map technology. In the case of subadditivity, $S(A)+S(B)\geq S(AB)$, the cuts that separate $A$ from its complement and the cuts that separate $B$ from its complement will by definition separate $AB$ from its complement, and not necessarily optimally. This proves subadditivity. For strong subadditivity, $S(AB)+S(BC)\ge S(B)+S(ABC)$, the argument is much the same: the loop cuts that separate $AB$ from its complement and $BC$ from its complement can also be used to separate $B$ from its complement and $ABC$ from its complement. In addition, the subsets of loops of the $AB$ and $BC$ cuts that get recycled to build loop cuts for each $B$ and $ABC$ can be chosen to be disjoint, thus preventing double-counting on the RHS without accompanying double-counting on the LHS. The converse inequality is clearly false: loops in cuts for $B$ and $ABC$ do not always allow one to build a loop cut for e.g. $AB$, as in particular $A$ can still be linked with $C$. This type of reasoning is very much the same one involved in the proof by contraction.}

\update{The above arguments show that the entropy cone of link models is contained within the $n$-party quantum cone for $n\leq3$. We could show containment in the other direction by exhibiting links that realize the extreme rays of the corresponding cones. This is immediate, since any hypergraph entropies can be realized by links as explained in Figure~\ref{fig:link_conversion}, and it is known that hypergraphs realize all such extreme rays \cite{Bao:2020zgx}. We hence find that entropy cone of link models is equivalent to the quantum entropy cone for $n \leq 3$ parties.}

\subsection{Potential Alternative Formulation: Super-Loops}
Here, we briefly discuss a more natural generalization of the contraction-map method to loops, in the sense that it better captures how (hyper)graphs work. We are unable to prove the statement, so we leave it as a conjecture.

An intuitive way to understand the $W(x)$ sets in the graph case is as ``super-vertices'' -- because one is only interested in the edges that enter and leave $W(x)$, the connections within a given $W(x)$ are unimportant, and one can think of each $W(x)$ as simply being a large vertex connected to other super-vertices by appropriate edges.  Indeed, it is precisely this intuition that is used in the proofs of polyhedrality of the graph and hypergraph entropy cones~\cite{Bao:2015bfa,Bao:2020zgx}. 

One would like to generalize this to the loop case, letting $W(x)$ denote a ``super-loop'', and redefining the notion of a $k$-bridge to correspond to a set of $k$-many $W(x)$ sets. This would immediately simplify the proof and notation of the contraction condition above: we would no longer need to separately sum over $n$-many $W(x)$ sets to find all possible $k$-bridges. Instead, we would simply have to sum over all $k$-many subsets of the $W(x)$ sets directly.

Unfortunately, the barrier to doing so lies in, once again, the fact that there is no longer a unique way to connect loops. The key in applying the super-vertex method is that cut entropies on the new graph involving super-vertices are identical to those on the old graph. This is immediate on hypergraphs due to the simplicity of hyperedges. However, it is not clear whether the same is true for loop models. The freedom and variety of ways to connect loop sets together means that the preservation of cut entropies in the new super-loop model is not obviously guaranteed. We leave it as an open question whether such a generalization is possible.

\section{Discussion}\label{sec:disc}
In this work, we extend the program of building models of multipartite entanglement based on min-cuts initiated in~\cite{Bao:2015bfa} for graphs and in~\cite{Bao:2020mqq, Bao:2020zgx, Walter:2020zvt} for hypergraphs to topological links. In particular, we motivate a specific topological generalization of the min-cut to a model of entanglement involving links where the entanglement entropy of subsystems is given by weights of linked loops. We have demonstrated the usefulness of this definition by showing that link models allow for the construction of entropy vectors which violate known inequalities for hypergraph states-- see, for example, the link model in Figure \ref{fig:link_15}. Additionally, we generalized the contraction-map method for proving inequalities, although this required the inclusion of certain oracular functions.
The inherent difficulty of working with the topological properties of these link models of entanglement leaves many open directions for future work.

\update{Firstly}, a more direct version of the contraction-map method with less reliance on oracles would be desirable in order to turn the proof method proposed here into a combinatorial one that can be implemented by computer. Likely this would involve formulating the necessary conditions in a cut-independent way. Alternatively, were this not to be possible, it would be interesting to prove that oracles are necessary and thus that this method just cannot be easily turned into a combinatorial one.

\update{Relatedly, while we are able to show that link models reach beyond the hypergraph cone, we are unable to obtain insight into the geometric structure of the link cone, e.g. we do not know if the entropy cone of link models is polyhedral. It is possible that insights into the full structure of the cone may enable a more efficient or direct proof approach. We remark that lack of polyhedrality is not necessarily an issue from a conceptual standpoint: it is unknown if the quantum entropy cone is polyhedral for $n > 3$ parties, and it may well be that it is not.}

\update{Next}, a heuristic strategy of constructing more ``interesting'' topological link models that are provably not realizable by hypergraphs would be desirable. This would require a more complete understanding of how to produce general topological links whose entropies, like those of the link in Figure \ref{fig:link_15}, go beyond the limits of the hypergraph paradigm.
Ultimately, one would hope to find a method for, given an entropy vector, determining whether there exists a link model realizing it and, if so, which form it takes. This would precisely correspond to a generalization of the powerful methods developed in~\cite{Avis:2021xnz} for graphs (also used in~\cite{Bao:2020zgx,Bao:2020mqq} for hypergraphs) via integer linear programs. For link models, it is possible that nascent machine learning methods in topological link/knot classification \`a la those of~\cite{Gukov:2020qaj} will be useful in heuristically generating these nontrivial link structures.

\update{Moreover, even if there does exist a link model realizing an entropy vector, it may not be the only one -- there may be multiple inequivalent link representations of a given entropy vector. This is similar to (hyper)graph models, where it is also possible to find arbitrarily many inequivalent graph representations of a given entropy vector. Hence, we do not believe that this non-uniqueness impacts the comparison of the link model entropy cone with the stabilizer and quantum entropy cones. It would, of course, be interesting to find ``minimal'' links for each entropy vector, analogous to the notion of minimality for graph representatives of entropy vectors defined by \cite{Avis:2021xnz}. Rigorously defining a sharp notion of minimality on links would certainly be useful for distilling which properties of links characterize entanglement structures unreachable by hypergraphs. Along similar lines, it would be interesting to determine what, if any, implications the extra information contained in inequivalent representations of the same entropy vector has for physical quantum states.}

\update{Finally, and perhaps most importantly, we would like to further strengthen the connections between topological link models and physical quantum state entropies. As argued above, topological link models obey the only two inequalities known to hold on arbitrary quantum states: subadditivity and strong subadditivity (and symmetry permutations thereof). Though necessary as conditions for topological link models to capture precisely the entropies of quantum states, these are not sufficient. The following logically independent questions remain:
\begin{itemize}
    \item \emph{Do topological link models obey all other quantum inequalities?} If the answer to this were to be affirmative, its proof would be of paramount importance, and correspondingly a very ambitious challenge. The reason for this is that other quantum inequalities are strongly believed to exists, but none are actually known! Obtaining an inequality valid on link models through the machinery presented in this work would shed very exciting light on this. If instead the answer were to be negative, the finding of a counterexample would be extremely interesting: it would provide a set of link entropies consistent with subadditivity and strong subadditivity, yet inconsistent with quantum mechanics, thereby pointing at the existence of novel quantum entropy inequalities that rule out such entropies.
    \item \emph{Do topological link models reproduce all quantum entropies?} If true, a proof of this in the affirmative would require either an explicit construction of link models for given quantum entropy vectors, or a proof of nonexistence of valid link inequalities stronger than quantum ones, both of which seem beyond the reach of current tools. If false, however, it would only take the proof of some inequality on links that does not hold quantum mechanically, which may be more achievable using the tools in this paper. 
\end{itemize}
As can be seen, the landscape of possibilities is vast and worth exploring. A YY/YN/NY/NN answer to the questions above would imply that the link entropy cone is equal to/a proper subset of/a proper superset of/incomparable to the quantum entropy cone. While the YY option would be particularly tantalizing, we believe that any of them would be of interest in advancing our tools for classifying multipartite quantum entanglement. If YN, link models would continue the program started in the context of holography in the obvious way: the most basic patterns of multipartite entanglement would be those realizable by graphs, then those realizable by hypergraphs but not graphs, then those realizable by topological links but not (hyper)graphs, etc. In fact, if NN, the same classification would still apply, since as explained around Figure~\ref{fig:link_conversion} all hypergraph entropies are realizable by links. Of course, in this latter NN case there would however be ``unphysical'' link models realizing entropies which are not quantum mechanical. One may nonetheless expect that some structural constraint on link models could be implemented to avoid this, potentially bringing one back to the interesting YN case. In the NY case, an analogous constraint would also be very interesting, as it could return one to the ideal YY case.
}

\update{
Let us emphasize why such a classification of multipartite entanglement into graphs, hypergraphs, links, etc. is particularly exciting. Graph models of entanglement originated naturally in the context of holography \cite{Bao:2015bfa}, where the entanglement of boundary quantum states is geometrized in the bulk in a rather explicit way. This geometrization is possible precisely because of the local character of quantum correlations in holographic states. We may thus characterize quantum states having entropies realizable by graphs as possessing only \emph{locally} entangled subsystems. Hypergraph models break this notion of locality by allowing for $k$-local correlations, implemented in the form of $k$-edges~\cite{Bao:2020zgx}. As a result, hypergraphs are further capturing quantum states which may also involve subsystems which have \emph{non-locally} entanglement in a multi-local sense. The additional extension to links brings us quite naturally to another layer of complexity in the structural form of entanglement where multi-locality is superseded by a more global notion of \emph{topological} entanglement. We thus see link models as an appealing extension of the classification of multipartite entanglement initiated in the context of holography in terms of graphs and followed by hypergraphs.
}


Given the promise demonstrated by the topological link model in going beyond the entropies describable by hypergraphs, and also given the large range of open questions for future exploration in this direction, we anticipate that this will be an open area of exploration in the area of multipartite entanglement for years to come.

\section*{Acknowledgements}
We would like thank Aaron Elersich, Greg Kuperberg, Jonathan Harper, Nils Baas, and Matthew Headrick for useful discussions. 
For parts of the completion of this work, N.B. was supported by the Department of Energy under grant number DE-SC0019380, and is supported by the Computational Science Initiative at Brookhaven National Laboratory, and by the U.S. Department of Energy QuantISED  Quantum Telescope award. 
N.C. and V.P.S. were supported in part by the National Science Foundation (award number PHY-1521446), and by the U.S. Department of Energy under contract DE-AC02-05CH11231 and award DE-SC0019380. 
V.P.S. gratefully acknowledges support by the NSF Graduate Research Fellowship Program under Grant No. DGE 1752814.
During the completion of this work, S.H.C. was supported by NSF grant PHY-1801805 and by funds from the University of California Santa Barbara, and is supported by NSF grant PHY-2107939 and by a Len DeBenedictis Graduate Fellowship.

\ignore{
Hence we can write
\begin{equation}
    C_l^{(k)} = \left\{ \ell\in\mathcal{L} \,:\, \mathcal{B}_l(x^{(1)},\dots,x^{(k)}) \right\}_{x^{(1)},\dots,x^{(k)}},
\end{equation}
where $\mathcal{B}_l$ is some Boolean oracle function that we wish to break down further. Letting $B'=\bigcup_{i=1}^{k} W(x^{(i)})$, the oracle function is asked to yield
\begin{equation}
    \mathcal{B}_l(x^{(1)},\dots,x^{(k)}) = 
    \begin{cases}
        \text{True} &\text{if $B'$ is a $k$-bridge for $I_l$,} \\
        \text{False} \quad &\text{otherwise.}
    \end{cases}
\end{equation}
Notice that one has $B'\cap C_l=\{\ell\}$ whenever the oracle yields True. An obvious necessary condition for $\mathcal{B}_l$ to yield true is that, among the bitstrings $k$ bitstrings $x^{(i)}$ going into its argument, at least one and at most $k-2$ have bit $x^{(i)}_l=1$. Respectively, this is because the bridge must intersect $W_l$, but must also have at least $2$ loops outside so as to intersect both $C_l$ and $\overline{W}_l$.}

\bibliographystyle{JHEP}
\bibliography{references}

\begin{thebibliography}{10}

\bibitem{Ryu:2006bv}
Shinsei Ryu and Tadashi Takayanagi.
\newblock ``{Holographic derivation of entanglement entropy from AdS/CFT}''.
\newblock \href{https://dx.doi.org/10.1103/PhysRevLett.96.181602}{Phys. Rev.
  Lett. {\bf 96}, 181602}~(2006).
\newblock  \href{http://arxiv.org/abs/hep-th/0603001}{arXiv:hep-th/0603001}.

\bibitem{Bao:2015bfa}
Ning Bao, Sepehr Nezami, Hirosi Ooguri, Bogdan Stoica, James Sully, and Michael
  Walter.
\newblock ``{The Holographic Entropy Cone}''.
\newblock \href{https://dx.doi.org/10.1007/JHEP09(2015)130}{JHEP {\bf 09},
  130}~(2015).
\newblock  \href{http://arxiv.org/abs/1505.07839}{arXiv:1505.07839}.

\bibitem{HernandezCuenca:2019wgh}
Sergio Hern\'andez-Cuenca.
\newblock ``{Holographic entropy cone for five regions}''.
\newblock \href{https://dx.doi.org/10.1103/PhysRevD.100.026004}{Phys. Rev. D
  {\bf 100}, 026004}~(2019).
\newblock  \href{http://arxiv.org/abs/1903.09148}{arXiv:1903.09148}.

\bibitem{Avis:2021xnz}
David Avis and Sergio Hern\'andez-Cuenca.
\newblock ``{On the foundations and extremal structure of the holographic
  entropy cone}''~(2021).
\newblock  \href{http://arxiv.org/abs/2102.07535}{arXiv:2102.07535}.

\bibitem{Bao:2020zgx}
Ning Bao, Newton Cheng, Sergio Hern\'andez-Cuenca, and Vincent~P. Su.
\newblock ``{The Quantum Entropy Cone of Hypergraphs}''.
\newblock \href{https://dx.doi.org/10.21468/SciPostPhys.9.5.067}{SciPost Phys.
  {\bf 9}, 067}~(2020).
\newblock  \href{http://arxiv.org/abs/2002.05317}{arXiv:2002.05317}.

\bibitem{pippenger2003inequalities}
Nicholas Pippenger.
\newblock ``The inequalities of quantum information theory''.
\newblock \href{https://dx.doi.org/10.1109/TIT.2003.809569}{IEEE Transactions
  on Information Theory {\bf 49}, 773--789}~(2003).

\bibitem{Linden:2013kal}
Noah Linden, Franti\v{s}ek Mat\'u\v{s}, Mary~Beth Ruskai, and Andreas Winter.
\newblock ``{The Quantum Entropy Cone of Stabiliser States}''.
\newblock \href{https://dx.doi.org/10.4230/LIPIcs.TQC.2013.270}{LIPIcs {\bf
  22}, 270--284}~(2013).
\newblock  \href{http://arxiv.org/abs/1302.5453}{arXiv:1302.5453}.

\bibitem{Walter:2020zvt}
Michael Walter and Freek Witteveen.
\newblock ``{Hypergraph min-cuts from quantum entropies}''.
\newblock \href{https://dx.doi.org/10.1063/5.0043993}{J. Math. Phys. {\bf 62},
  092203}~(2021).
\newblock  \href{http://arxiv.org/abs/2002.12397}{arXiv:2002.12397}.

\bibitem{Nezami:2016zni}
Sepehr Nezami and Michael Walter.
\newblock ``{Multipartite Entanglement in Stabilizer Tensor Networks}''.
\newblock \href{https://dx.doi.org/10.1103/PhysRevLett.125.241602}{Phys. Rev.
  Lett. {\bf 125}, 241602}~(2020).
\newblock  \href{http://arxiv.org/abs/1608.02595}{arXiv:1608.02595}.

\bibitem{Bao:2020mqq}
Ning Bao, Newton Cheng, Sergio Hern\'andez-Cuenca, and Vincent~Paul Su.
\newblock ``{A Gap Between the Hypergraph and Stabilizer Entropy
  Cones}''~(2020).
\newblock  \href{http://arxiv.org/abs/2006.16292}{arXiv:2006.16292}.

\bibitem{Salton:2016qpp}
Grant Salton, Brian Swingle, and Michael Walter.
\newblock ``{Entanglement from Topology in Chern-Simons Theory}''.
\newblock \href{https://dx.doi.org/10.1103/PhysRevD.95.105007}{Phys. Rev. D
  {\bf 95}, 105007}~(2017).
\newblock  \href{http://arxiv.org/abs/1611.01516}{arXiv:1611.01516}.

\bibitem{Balasubramanian:2016sro}
Vijay Balasubramanian, Jackson~R. Fliss, Robert~G. Leigh, and Onkar Parrikar.
\newblock ``{Multi-Boundary Entanglement in Chern-Simons Theory and Link
  Invariants}''.
\newblock \href{https://dx.doi.org/10.1007/JHEP04(2017)061}{JHEP {\bf 04},
  061}~(2017).
\newblock  \href{http://arxiv.org/abs/1611.05460}{arXiv:1611.05460}.

\bibitem{Chun:2017hja}
Sungbong Chun and Ning Bao.
\newblock ``{Entanglement entropy from SU(2) Chern-Simons theory and symmetric
  webs}''~(2017).
\newblock  \href{http://arxiv.org/abs/1707.03525}{arXiv:1707.03525}.

\bibitem{Mironov:2019taf}
Sergey Mironov.
\newblock ``{Topological Entanglement and Knots}''.
\newblock \href{https://dx.doi.org/10.3390/universe5020060}{Universe {\bf 5},
  60}~(2019).

\bibitem{Kauffman:2019}
Louis~H. Kauffman and Eshan Mehrotra.
\newblock ``{Topological aspects of quantum entanglement}''.
\newblock \href{https://dx.doi.org/10.1007/s11128-019-2191-z}{Quantum Inf
  Process{\bf 18}}~(2019).

\bibitem{Aharonov2008APQ}
D.~Aharonov, V.~Jones, and Zeph Landau.
\newblock ``A polynomial quantum algorithm for approximating the jones
  polynomial''.
\newblock \href{https://dx.doi.org/10.1007/s00453-008-9168-0}{Algorithmica {\bf
  55}, 395--421}~(2006).

\bibitem{Akers:2021lms}
Chris Akers, Sergio Hern\'andez-Cuenca, and Pratik Rath.
\newblock ``{Quantum Extremal Surfaces and the Holographic Entropy Cone}''.
\newblock \href{https://dx.doi.org/10.1007/JHEP11(2021)177}{JHEP {\bf 11},
  177}~(2021).
\newblock  \href{http://arxiv.org/abs/2108.07280}{arXiv:2108.07280}.

\bibitem{Hein_2004}
M~Hein, Jens Eisert, and Hans Briegel.
\newblock ``Multiparty entanglement in graph states''.
\newblock \href{https://dx.doi.org/10.1103/PhysRevA.69.062311}{Phys. Rev. A
  {\bf 69}, 062311}~(2004).

\bibitem{habegger1990classification}
Nathan Habegger and Xiao-Song Lin.
\newblock ``The classification of links up to link-homotopy''.
\newblock \href{https://dx.doi.org/10.1090/S0894-0347-1990-1026062-0}{Journal
  of the American Mathematical SocietyPages 389--419}~(1990).

\bibitem{Gukov:2020qaj}
Sergei Gukov, James Halverson, Fabian Ruehle, and Piotr Su\l{}kowski.
\newblock ``{Learning to Unknot}''.
\newblock \href{https://dx.doi.org/10.1088/2632-2153/abe91f}{Mach. Learn. Sci.
  Tech. {\bf 2}, 025035}~(2021).
\newblock  \href{http://arxiv.org/abs/2010.16263}{arXiv:2010.16263}.

\end{thebibliography}

\end{document}